\documentclass[a4paper,twocolumn,11pt,unpublished]{quantumarticle}
\pdfoutput=1
\usepackage[english]{babel}
\usepackage[numbers]{natbib}
\usepackage{amsthm}
\usepackage{graphicx}
\usepackage{subcaption}
\usepackage{tikz}
\usetikzlibrary{quantikz}
\usepackage{Schritt}
\usepackage{dsfont}
\newtheorem{lemma}{Lemma}

\newtheorem{example}{Example}
\newtheorem{simplificationrule}{Rule}

\usepackage{bbding}
\def\1{\mathds{1}}
\newcommand{\e}{\mathrm{e}}

\begin{document}
	\title{Hybrid simplification rules for boundaries of quantum circuits}
	\author{Michael Epping}
	\orcid{0000-0003-0950-6801}
	\email{Michael.Epping@dlr.de}
	\affiliation{German Aerospace Center (DLR), Linder Höhe, 51147 Cologne, Germany}
	
	\begin{abstract}
		We describe rules to simplify quantum circuits at their boundaries, i.e. at state preparation and measurement. 
		There, any strictly incoherent operation may be pushed into a pre- or post-processing of classical data.
		The rules can greatly simplify the implementation of quantum circuits and are particularly useful for hybrid algorithms on noisy intermediate-scale quantum hardware,
		e.g. in the context of quantum simulation.
		Finally we illustrate how circuit cutting can enable the rules to be applied to more locations.
	\end{abstract}
	\maketitle
	
Near-term quantum computers do not have the capabilities for full-fledged error correction. 
Thus they suffer from noise and can only run circuits of limited size. 
Nevertheless, useful computations might be performed on such devices. 
In this so called noisy intermediate-scale quantum (NISQ) regime, 
it is crucial to optimize the quantum circuits to get meaningful results. 
The literature on this important building-block of a quantum software stack is steadily increasing.
It includes work on compilation techniques~\cite{Lomont2003,Maslov2008,Iten2019,Khatri2019}, which can also be quite hardware specific~\cite{Martinez2016}.
There are also several approaches to optimizing quantum circuits, e.g. w.r.t. the circuit depth or the controlled-not or T-gate count~\cite{Heyfron2018}.
For short circuits on few qubits optimal solutions can be found by meet-in-the-middle algorithms~\cite{Amy2013}.
However, larger circuits can only be optimized using heuristics.
Circuits can be simplified via a detour over ZX diagrams~\cite{Coecke2011,Kissinger2020,Backens2021}.
But often rewrite rules are used to simplify the circuit step by step~\cite{Nam2018,Pointing2021}.

The present work can be viewed to follow the latter approach as well.
The focus is slightly different, though, because we are interested in hybrid simplification rules.
With ``hybrid'' we mean rules which affect both the quantum and the classical part of a computation.
For NISQ applications it is quite natural to think about hybrid simplification rules, 
because the most promising applications for NISQ devices are hybrid anyways.
That is they use the quantum computer only for those parts of the algorithm, 
where it is actually beneficial.
An important class of such algorithms are variational quantum algorithms~\cite{Cerezo2021}. 

Before we collect hybrid simplification rules, 
we develop some intuition about what kind of operations we want to shift from the quantum to the classical computer. 
For each quantum computing hardware there is a preferred basis of the Hilbert space, the computational basis, $\ket{0}$, $\ket{1}$, ..., and $\ket{2^n-1}$.
States which are superpositions of such basis states are said to contain coherence~\cite{Baumgratz2014, Streltsov2017}.
And operations that map incoherent states to incoherent states are called incoherent operations.
These are considered free operations in the resource theory of coherence.
The possibility of superpositions is maybe the most crucial difference between classical and quantum computation.
In this sense the computational states are the classical states of the quantum computer.
They can also be stored and manipulated on a classical computer.
Given the noise on a near-term intermediate-scale quantum (NISQ) computer, operations that can be performed (efficiently) on a classical computer can indeed be considered free to some extent. 
It is quite fitting that the free operations in the resource theory of coherence are operations that can be performed on a classical computer, 
such that they are also ``free'' in the NISQ setting.
This leads to the idea that incoherent operations should be performed on a classical computer whenever possible.

In this paper we discuss circuit simplification rules that can be used to achieve exactly that.
We start by introducing some notation in Sectin~\ref{sec:notation}.
Afterwards we first look at unitary operations in Section~\ref{sec:unitary} and move on to general quantum operations in Section~\ref{sec:nonunitary}. Finally, in Section~\ref{sec:cutting} we discuss how circuit cutting can increase the usefulness of our approach and we conclude the paper in Section~\ref{sec:conclusion}.

\section{Notation}\label{sec:notation}
Without loss of generality, we will assume that a quantum algorithm starts with an incoherent state and ends with a measurement in the computational basis.
Thus all measurements will be in the computational basis unless explicitly stated otherwise.
The computational basis state
\begin{equation}
	\ket{x} = \bigotimes_{k=1}^{n} \ket{x_k}
\end{equation}
can be conveniently interpreted as a product state of single qubit states $\ket{x_k}$ where $x_k$ is the $k$-th binary digit of $x$.
We use commata to separate the binary digits of a number, e.g. $x=x_1,x_2,x_3$. 
Note that we reverse the order of the digits to match the usual endianness of qubit registers. 

In our circuit diagrams single, double, and triple wires denote qubits, classical information, and an arbitrary number of qubits, respectively. 
We use $\1$ to denote the identity matrix of appropriate dimension.
$X$, $Y$, and $Z$, denotes the Pauli matrices and $r_x$, $r_y$ and $r_z$ denotes the corresponding rotation operator gates.
A general controlled gate is denoted by
\begin{equation}
	cU_i := \sum_i \proj{i}\otimes U_i, \label{eq:cUi}
\end{equation}
where the $U_i$ are unitary.
More gates are listed in Table~\ref{tab:ppgates}.

\section{Unitary evolution}\label{sec:unitary}
In general incoherent operations are not unitary. 
However, we postpone the non-unitary case to Section~\ref{sec:nonunitary} and focus on the unitary case for now.
A unitary incoherent operation $U$ by definition maps any canonical basis state to a canonical basis state.
Via a natural generalization of the usual definitions, we express such unitary incoherent operations in the phase polynomial representation~\cite{Amy2013, Nam2018}.
Then we can write
\begin{equation}
	U \ket{x} = \e^{i p(x)} \ket{f(x)} \label{eq:phasepolynomialrepresentation}
\end{equation}
for all $ x\in\{0,1,...,2^n-1\}$, where we call $p$ the phase polynomial and $f$ the basis transformation.
In short we call $(p, f)$ the phase polynomial representation of $U$. 
We should note that we do not require that $p$ is actually a polynomial in the variable $x$, though.
In the present context $p(x)$ could as well be an arbitrary real function of $x$.
Examples of gates are given in Table~\ref{tab:ppgates}.
\begin{table*}
\caption{Examples of gates with phase polynomial representations.}
\label{tab:ppgates}
\centering
\begin{tabular}{c|c|c|c}
Name & Definition & Phase polynomial & Basis transformation\\
\hline
$X$ & $\ket{1}\bra{0}+\ket{0}\bra{1}$ & $p(x) = 0$ & $f(x) = x \oplus 1$ \\
$Y$ & $i \ket{1}\bra{0}-i\ket{0}\bra{1}$ & $p(x) = \pi x$ & $f(x) = x \oplus 1$ \\
$Z$ & $\proj{0}-\proj{1}$ & $p(x) = \pi x$ & $f(x) = x$\\
$T$ & $\proj{0} + \e^{i \pi/4} \proj{1}$ & $p(x) = \pi x /4$ & $f(x)=x$\\
$r_z$ & $\e^{i \theta/2}\proj{0}+\e^{-i\theta/2}\proj{1}$ & $p(x)=\theta x$ & $f(x)=x$ \\
$c_X$ & $\proj{0}\otimes \1 + \proj{1}\otimes X$ & $p(x)=0$ & $f(x_1,x_2) = x_1,x_1\oplus x_2$\\
Toffoli & $\1+\proj{1}^{\otimes 2}\otimes (X-\1)$ & $p(x)=0$ & $f(x_1,x_2,x_3) = x_1,x_2,x_3\oplus x_1x_2$
\end{tabular}
\end{table*}
The composition of gates into circuits is made explicit in the following lemma.
\begin{lemma}[Phase polynomial composition]
The composition of two circuits $U_1$ and $U_2$ with phase polynomial representations $(p_1, f_1)$ and $(p_2, f_2)$, respectively, has the phase polynomial representation $(p_1 + p_2\circ f_1, f_2 \circ f_1)$.	
\end{lemma}
\begin{proof}
	\begin{alignat*}{2}
		U_2 U_1 \ket{x} =& U_2 \e^{i p_1(x)} \ket{f_1(x)}\\
		=& \e^{i (p_1(x) + p_2(f_1(x)))} \ket{f_2(f_1(x))}
	\end{alignat*}
\end{proof}
This implies that a circuit which consists only of gates with a phase polynomial representation  can itself be expressed in this form.
A simple algorithm to convert a circuit given as a list of gates into its phase polynomial representation follows directly from this observation:
Start with the identity operation with phase polynomial representation $(0,x)$. Then update the basis transformation and the phase polynomial gate by gate.

The following observation allows to simplify circuits with phase polynomials on their boundary.
\begin{simplificationrule}[Phase polynomial trimming]\label{rule:pptrimming}
Let $F$ be a quantum circuit with phase polynomial representation $(p, f)$, i.e. $F\ket{x}=\e^{i p(x)} \ket{f(x)}$. Then $F$ can be replaced by classical pre- and post-processing via 
\begin{alignat}{2}
	\raisebox{4pt}{\begin{tikzcd}[column sep=4mm]
			\lstick{\ket{x}}& \gate{F}\qwbundle[alternate]{} & \qwbundle[alternate]{}
	\end{tikzcd}}=&
	\raisebox{4pt}{\begin{tikzcd}[column sep=4mm]
			\lstick{\ket{f(x)}}&\qwbundle[alternate]{}
	\end{tikzcd}}\label{eq:trimppleft}
	\intertext{and}
	\raisebox{4pt}{\begin{tikzcd}[column sep=4mm]
			\qwbundle[alternate]{}&\gate{F}\qwbundle[alternate]{} & \meter{}\qwbundle[alternate]{}
	\end{tikzcd}}=&
	\raisebox{4pt}{\begin{tikzcd}[column sep=4mm]
			\qwbundle[alternate]{}& \meter{}\qwbundle[alternate]{} &\gate[cwires={1}]{f}
	\end{tikzcd}}, \label{eq:trimppright}
\end{alignat}
respectively. 
\end{simplificationrule}
The practical relevance of Rule~\ref{rule:pptrimming} is that any circuit which starts or ends with a part that allows a phase polynomial representation can be simplified by classical pre- or post-processing.
The following example illustrates this application.
\begin{example}
	The rule in Eq.~(\ref{eq:trimppright}) can be applied as in
	
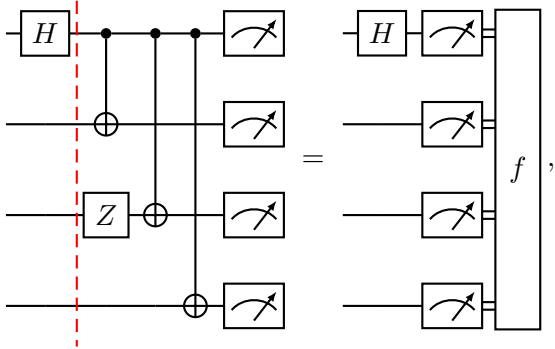
\begin{equation*}
	\begin{tikzcd}[column sep=2mm]
		\qw & \gate{H}\slice{} & \ctrl{1} & \ctrl{2} & \ctrl{3} & \meter{}\\
		\qw & \qw & \targ{} & \qw & \qw & \meter{}\\
		\qw & \qw & \gate{Z} & \targ{} & \qw & \meter{}\\
		\qw & \qw & \qw & \qw & \targ{} & \meter{}
	\end{tikzcd}
=
\begin{tikzcd}[column sep=2mm]
	\qw & \gate{H} & \meter{} & \gate[4,cwires={1,2,3,4}]{f}\\
	\qw & \qw &  \meter{}&\\
	\qw & \qw &  \meter{}&\\
	\qw & \qw &  \meter{}&
\end{tikzcd},
	\end{equation*}
where $$f(x_1,x_2,x_3,x_4) = x_1,x_2\oplus x_1, x_3 \oplus x_1, x_4 \oplus x_1.$$
\end{example}
Sometimes rewriting, or re-compiling parts of the circuit, is necessary before one can trim a phase polynomial.
A typical situation is discussed in the following example.
\begin{example}
	Any single qubit unitary can be decomposed into the gate sequence $r_z r_x r_z$. 
	Note that $r_z$ has a phase polynomial representation, see Table~\ref{tab:ppgates}.
	Clever use of this decomposition allows to simplify a circuit as in
	\begin{equation}
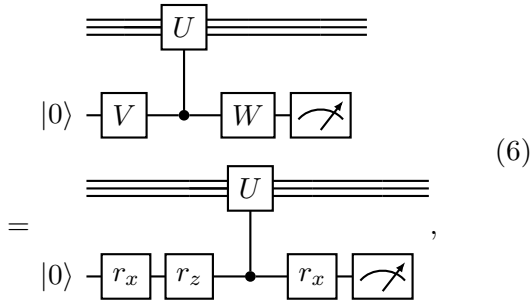
		\begin{aligned}
		&\begin{tikzcd}[column sep=2mm]
			\qwbundle[alternate]{}&\qwbundle[alternate]{}& \gate{U}\qwbundle[alternate]{} & \qwbundle[alternate]{}& \qwbundle[alternate]{}& \qwbundle[alternate]{}\\
			\lstick{$\ket{0}$}&\gate{V} & \ctrl{-1} & \gate{W} & \meter{}
		\end{tikzcd}\\
	=& 
	\begin{tikzcd}[column sep=2mm]
		\qwbundle[alternate]{}&\qwbundle[alternate]{}&\qwbundle[alternate]{}& \gate{U}\qwbundle[alternate]{} & \qwbundle[alternate]{}& \qwbundle[alternate]{}& \qwbundle[alternate]{}\\
		\lstick{$\ket{0}$}&\gate{r_x} & \gate{r_z} & \ctrl{-1} & \gate{r_x} & \meter{}
	\end{tikzcd},
	\end{aligned}
	\end{equation}
where we also used that $r_z$ commutes with the controlled-$U$ gate on the control qubit.
\end{example}


Further possible simplifications are contained in the following rule.
\begin{simplificationrule}\label{rule:rewrite}
	Given a unitary $U$, one can choose a unitary $V$ such that $F:=V^{\dagger} U$ has phase polynomial representation $(p, f)$. Then
	\begin{alignat}{2}
		\raisebox{4pt}{\begin{tikzcd}[column sep=2mm]
			\qwbundle[alternate]{}&\gate{U}\qwbundle[alternate]{} & \meter{}\qwbundle[alternate]{}
		\end{tikzcd}} =& \raisebox{4pt}{\begin{tikzcd}[column sep=2mm]
	\qwbundle[alternate]{}&\gate{V}\qwbundle[alternate]{} & \meter{}\qwbundle[alternate]{} & \gate[cwires={1}]{f}
	\end{tikzcd}}
\intertext{and}
		\raisebox{4pt}{\begin{tikzcd}[column sep=2mm]
			\lstick{\ket{x}}&\gate{U}\qwbundle[alternate]{} & \qwbundle[alternate]{}
	\end{tikzcd}} =& \raisebox{4pt}{\begin{tikzcd}[column sep=2mm]
			\lstick{\ket{f(x)}}&\gate{V}\qwbundle[alternate]{} & \qwbundle[alternate]{} 
	\end{tikzcd}}.
\end{alignat}
\end{simplificationrule}
Rule~\ref{rule:rewrite} is useful, if $V$ is easier to implement than $U$. 
We are mostly interested in the following special case. 
The gate set ``Clifford + T''~\cite{Boykin2000}, $\{H, c_X, T\}$, is very common for fault-tolerant quantum computation. 
This gate set could also be called ``Incoherent Operations + H''. 
If we consider circuits which solely consist of $r_z$ (which includes $T$ as special cases), $c_x$ and $H$, then the Hadamard gates form obstacles for simplifications using Rule~\ref{rule:pptrimming}.
However, sometimes it is still possible to push phase polynomials past such an obstacle.
This is best visualized with the example of a controlled-Not gate.
\begin{example}
	For $U=c_x H^{\otimes 2}$ and $V=H^{\otimes 2}$ Rule~\ref{rule:rewrite} gives
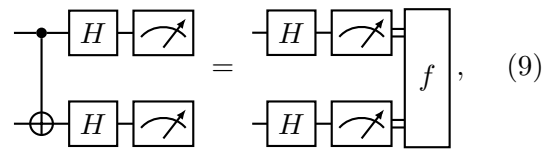
	\begin{equation}
	\begin{tikzcd}[column sep=2mm]
		\qw&\ctrl{1} &\gate{H} & \meter{}\\
		\qw&\targ{} & \gate{H} & \meter{}
	\end{tikzcd} = \begin{tikzcd}[column sep=2mm]
	\qw&\gate{H} & \meter{} & \gate[2,cwires={1,2}]{f}\\
	\qw&\gate{H} & \meter{} & 
\end{tikzcd},
	\end{equation}
with $f(x_1,x_2)=x_1 \oplus x_2, x_2$.
\end{example}

It is well-known that controlled unitary gates before a measurement in the computational basis can be replaced by classically controlled gates. 
We include this rule here for completeness. 
\begin{simplificationrule}[Classically-controlled gates]\label{rule:classicallycontrolled}
	Let the general controlled-$U_i$ operation be defined as in Eq.~(\ref{eq:cUi}). If the control-qubits are measured in the computational basis directly after this gate, then it can be replaced by a classically-controlled gate: 
	\begin{equation}
		\begin{tikzcd}
			&\ctrlbundle{1} & \meter{}\qwbundle[alternate]{}\\
			&\gate{U_i}\qwbundle[alternate]{} & \qwbundle[alternate]{}
		\end{tikzcd} = \begin{tikzcd}
			& \meter{}\qwbundle[alternate]{} \\
			&\gate{U_i}\vcw{-1}\qwbundle[alternate]{} & \qwbundle[alternate]{} 
		\end{tikzcd}.
	\end{equation}
For a controlled gate directly after the preparation of a computational basis state on the control qubits the analogous equation
\begin{equation}
	\begin{tikzcd}
		\lstick{$\ket{x}$}&\ctrlbundle{1} & \qwbundle[alternate]{} \\
		\qwbundle[alternate]{}&\gate{U_i}\qwbundle[alternate]{} & \qwbundle[alternate]{}
	\end{tikzcd} = \begin{tikzcd}
\lstick{$\ket{x}$} & \qwbundle[alternate]{} &\qwbundle[alternate]{} \\
		&\gate{U_x}\qwbundle[alternate]{} &\qwbundle[alternate]{}
	\end{tikzcd}
\end{equation}
holds.
\end{simplificationrule}
The following example illustrates how Rule~\ref{rule:classicallycontrolled} can be used to simplify the quantum Fourier transform of a fixed input state.
\begin{example}
	The quantum Fourier transform (QFT)~\cite{Coppersmith2002} is a crucial component of many quantum algorithms which provide a speedup compared to their classical counter-parts. Sometimes it is convenient to write that the QFT is applied to $\ket{0}$, e.g. when followed by the inverse QFT later on. For $n=4$ qubits the circuit has the form
	\begin{widetext}
		\begin{equation}
			QFT\ket{0}^{\otimes 4} = \begin{tikzcd}[column sep=2mm]
				\lstick{\ket{0}} & \gate{H} & \gate{R_2} & \gate{R_3} & \gate{R_4} & \qw& \qw& \qw& \qw& \qw& \qw& \qw& \qw\\
				\lstick{\ket{0}} & \qw & \ctrl{-1} & \qw & \qw & \qw & \gate{H} & \gate{R_2} & \gate{R_3} & \qw& \qw& \qw& \qw\\
				\lstick{\ket{0}} & \qw & \qw & \ctrl{-2} &\qw & \qw &\qw & \ctrl{-1} &\qw & \gate{H} & \gate{R_2}& \qw& \qw\\
				\lstick{\ket{0}} & \qw & \qw & \qw & \ctrl{-3} & \qw &\qw & \qw & \ctrl{-2} & \qw & \ctrl{-1} & \gate{H}& \qw\\
			\end{tikzcd}.
		\end{equation}
	\end{widetext}
	One should note, however, that with Rule~\ref{rule:classicallycontrolled} all multi-qubit gates can be removed. So in this case $QFT\ket{0}=H^{\otimes 4} \ket{0}$.
\end{example}
\section{Non-unitary quantum operations}~\label{sec:nonunitary}
Up until now we shifted only unitary operations across the boundaries of the quantum circuit. The result was a deterministic pre- and post-processing.
We now extend the intuition that ``classical operations should be performed on the classical side of the boundary'' to non-unitary channels.

A general quantum operation on a density matrix $\rho$ can be expressed as
\begin{equation}
	\Phi(\rho) = \sum_n K_n \rho K_n^\dagger
\end{equation}
via its Kraus operators $K_n$ which fulfill
\begin{equation}
	\sum_n K_n^\dagger K_n = \1. 
\end{equation} 
If $K_n$ is incoherent for all $n$, then $\Phi$ is incoherent. If, additionally, the $K_n^\dagger$ are incoherent, then $\Phi$ is strictly incoherent (SI)~\cite{Yadin2016}.
While incoherent operations do not create coherence, SI operations additionally do not use any coherence of the input state.

Depending on the boundary type it suffices to require that the Kraus operators $K_n$ or their adjoints $K_n^\dagger$ are incoherent.
For example consider the $X$-measurement on a single qubit, which is not a SI operation.
It can be pushed to classical pre- but not post-processing.
However, for simplicity, we will consider SI operations only.  

SI operations can be interpreted as a classical stochastic process, which is described by a stochastic transition matrix $M$~\cite{Yadin2016}, with
\begin{equation}
	M_{ij} = \sum_n |\bra{i}K_n \ket{j}|^2.
\end{equation}
 The measurement probabilities transform like
\begin{equation}
	\begin{aligned}
	p_{i} \mapsto& \bra{i}\Phi(\rho)\ket{i}\\
	 =&\sum_j M_{ij} p_j.
	 \end{aligned}
\end{equation}

However, in the context of circuit optimization, we do not start from the Kraus operator representation.
It was shown in \cite{Yadin2016} that any SI operation can be written as an incoherent interaction with an environment.
The following simplification rule describes how any such SI operation on the boundary of the quantum circuit can be shifted to a classical, but non-deterministic, pre- or post-processing.

\begin{simplificationrule}[SI operations]\label{rule:incoherent}
For a circuit on $n$ system qubits and $n_E$ ancillary qubits,
let the controlled-$U_i$ gate be defined as in Eq.~(\ref{eq:cUi}) and let $F_\mu$ have phase polynomial representation $(p_\mu, f_\mu)$. 
The measurement outcomes are $x\in\{0,1,...,2^n-1\}$ and $\mu\in\{0,1,...,2^{n_E}-1\}$ on the system and the ancillas, respectively.
Then
\begin{equation}
	\begin{tikzcd}[column sep=2mm]
		\lstick{$\ket{0}$}&\gate{U_i}\qwbundle[alternate]{} & \meter{}\qwbundle[alternate]{}\vcw{1}&\\
		\qwbundle[alternate]{}&\ctrlbundle{-1} & \gate{F_\mu}\qwbundle[alternate]{} & \meter{}\qwbundle[alternate]{}
	\end{tikzcd}
=
	\raisebox{4pt}{\begin{tikzcd}[column sep=2mm]
		& \meter{}\qwbundle[alternate]{} & \gate[cwires={1}]{f}
	\end{tikzcd}},
\end{equation}
where the stochastic process $f$ is defined via the stochastic matrix
\begin{equation}
	M = \sum_{\mu,x} P(\mu | x) \ket{f_\mu(x)} \bra{x}
\end{equation}
with
\begin{equation}
	P(\mu|x) = |\bra{\mu}U_x\ket{0}|^2. \label{eq:Pmux}
\end{equation}
Analogously,
\begin{equation}
	\begin{tikzcd}[column sep=2mm]
		& \ket{0}&\gate{U_i}\qwbundle[alternate]{} & \meter{}\qwbundle[alternate]{}\vcw{1}&\\
		\lstick{$\ket{x}$}&\qwbundle[alternate]{}&\ctrlbundle{-1} & \gate{F_\mu}\qwbundle[alternate]{}&\qwbundle[alternate]{}
	\end{tikzcd}
	=
	\raisebox{4pt}{\begin{tikzcd}[column sep=2mm]
		 &\lstick{$\ket{f(x)}$}&[5mm] \qwbundle[alternate]{} 
	\end{tikzcd}}.
\end{equation}
\end{simplificationrule}
\begin{proof}
\begin{Schritte}
	\begin{tikzcd}[column sep=2mm]
		\lstick{$\ket{0}$}&\gate{U_i}\qwbundle[alternate]{} & \meter{}\qwbundle[alternate]{}\vcw{1}\slice{} &\\
		\qwbundle[alternate]{}&\ctrlbundle{-1} & \gate{F_\mu}\qwbundle[alternate]{} & \meter{}\qwbundle[alternate]{}
	\end{tikzcd}
\schritthier{\overset{Eq.~(\ref{eq:trimppright})}{=}}&
\begin{tikzcd}[column sep=2mm]
	\lstick{$\ket{0}$}&\gate{U_i}\qwbundle[alternate]{} & \meter{}\qwbundle[alternate]{}& \cwbend{1}\\
	\qwbundle[alternate]{}&\ctrlbundle{-1} & \meter{}\qwbundle[alternate]{} & \gate[cwires={1}]{f_\mu}
\end{tikzcd}\\
\schritttext{Apply Rule~\ref{rule:classicallycontrolled}.}
\schritthier{\hspace{3mm}=\hspace{3mm}}&
\begin{tikzcd}[column sep=2mm]
	\lstick{$\ket{0}$}&\gate{U_i}\qwbundle[alternate]{}\vcw{1} & \meter{}\qwbundle[alternate]{}\vcw{1}&\\
	\qwbundle[alternate]{}& \meter{}\qwbundle[alternate]{} & \gate[cwires={1}]{f_\mu}&
\end{tikzcd}\\
\schritttext{\begin{minipage}{3cm}\flushleft
		Randomly choose $\mu$ with probability $P(\mu|x)=|\bra{\mu}U_x\ket{0}|^2$
	\end{minipage}
}
\schritthier{\hspace{3mm}=\hspace{3mm}}&
\raisebox{4pt}{\begin{tikzcd}[column sep=2mm]
	& \meter{}\qwbundle[alternate]{} & \gate[cwires={1}]{f_{\mu}}
\end{tikzcd}}\\
\schritttext{\begin{minipage}{3cm}\flushleft
		Combine sampling of $\mu$ and application of $f_{\mu}$ into a single stochastic process $f$.
	\end{minipage}}
\schritthier{\hspace{3mm}=\hspace{3mm}}&
\raisebox{4pt}{\begin{tikzcd}[column sep=2mm]
		& \meter{}\qwbundle[alternate]{} & \gate[cwires={1}]{f}
\end{tikzcd}}
\end{Schritte}
\end{proof}
Note that the complexity of calculating $P(\mu|x)$ in general scales exponentially in the number of involved qubits. 
Thus Rule~\ref{rule:incoherent} can only be applied to subsystems with limited size.
The following example illustrates the use of Rule~\ref{rule:incoherent} for an ancillary system of size $n_E=2$.
\begin{example}
Consider the circuit
\begin{equation}
	C = \begin{tikzcd}[column sep=2mm]
		\lstick{$\ket{0}$}&\qw & \targ{}& \qw & \meter{} &  \cwbend{2}\\
		\lstick{$\ket{0}$}&\gate{H}& \ctrl{-1} & \gate{r_x\left(\frac{\pi}{2}\right)} & \meter{}\vcw{1} \\
		\lstick[2]{$\ket{\psi_{\mathrm{in}}}$}&\qw&\ctrl{-1}&\qw & \gate{X} & \gate[2]{c_X} &  \meter{}\\
		 & \qw & \qw & \ctrl{-2} & \qw & \qw & \meter{}
	\end{tikzcd}
\end{equation}
with a SI operation directly before a measurement in the computational basis.
We apply Rule~\ref{rule:incoherent} to turn the incoherent operation into a classical post-processing. The circuit becomes
\begin{equation}
	C =\raisebox{4pt}{ \begin{tikzcd}[column sep=2mm]
		\lstick[2]{$\ket{\psi_{\mathrm{in}}}$}&\meter{}&\gate[2,cwires={1,2}]{f}\\
		& \meter{} &
	\end{tikzcd}}.
\end{equation}
The interaction with the ancillary system is replaced by the stochastic process $f$, which acts as a post-processing step.
We identify $F_\mu$ with the classically controlled operations and determine the corresponding basis transformation
\begin{equation}
	\begin{aligned}
		f_{0,0}(x_1,x_2) =& x_1,x_2,\\
		f_{0,1}(x_1,x_2) =& x_1\oplus 1, x_2,\\
		f_{1,0}(x_1,x_2) =& x_1, x_1\oplus x_2,\\
	\text{and }	f_{1,1}(x_1,x_2) =& x_1\oplus 1, x_1 \oplus x_2\oplus 1.\\
	\end{aligned}
\end{equation}
The measurement probabilities $P(\mu|x)$ on the ancillary system,
\begin{equation}
	P(\mu|x) = \left(
	\begin{array}{cccc}
		\frac{1}{2} & \frac{1}{2} & \frac{1}{2} & \frac{1}{4} \\
		\frac{1}{2} & 0 & \frac{1}{2} & \frac{1}{4} \\
		0 & 0 & 0 & \frac{1}{4} \\
		0 & \frac{1}{2} & 0 & \frac{1}{4} \\
	\end{array}
	\right)_{\mu x},
\end{equation}
are calculated using Eq.~(\ref{eq:Pmux}).
Then the stochastic matrix is
\begin{equation}
	M=\left(
	\begin{array}{cccc}
		\frac{1}{2} & 0 & \frac{1}{2} & 0 \\
		0 & \frac{1}{2} & 0 & \frac{1}{2} \\
		\frac{1}{2} & 0 & \frac{1}{2} & \frac{1}{4} \\
		0 & \frac{1}{2} & 0 & \frac{1}{4} \\
	\end{array}
	\right).
\end{equation}
For
\begin{equation}
	\ket{\psi_{\mathrm{in}}} = \alpha \ket{0} + \beta \ket{1} + \gamma \ket{2} + \delta \ket{3}
\end{equation}
the measurement probabilities are transformed into
\begin{equation}
	T \left(
	\begin{array}{c}
		|\alpha|^2\\
		|\beta|^2\\
		|\gamma|^2\\
		|\delta|^2
	\end{array}\right)
	=	
	\frac{1}{4}\left(\begin{array}{c}
		2 |\alpha| ^2+2|\gamma|^2\\
	2 |\beta|^2+2 |\delta|^2\\
	 2 |\alpha|^2+2|\gamma|^2+|\delta|^2\\
	2 |\beta|^2+|\delta|^2
	\end{array}\right),
\end{equation}
which can also be verified by direct calculation.
\end{example}
\section{Circuit cutting enables more simplifications}\label{sec:cutting}
Circuit cutting refers to methods which allow to split a larger quantum circuit into smaller ones~\cite{Bravyi2016,Peng2020,Mitarai2021,Piveteau2022}. 
The measurement outcomes are recombined to simulate the larger circuit. 
We distinguish two different kinds of cuts: time- and spacewise. 
In the usual circuit diagrams with time flowing from left to right and the qubits arranged from top to bottom, time- and spacewise cuts are vertical and horizontal, respectively.
The problem of finding good positions for cuts is challenging, but can be tackled with mixed integer programming~\cite{Tang2021}.
The following lemma, taken from \cite{Mitarai2021}, describes a horizontal cut.
\begin{lemma}[Horizontal cut~\cite{Mitarai2021}]\label{lemma:horizontalcut}
	For operators $A_1$ and $A_2$ with $A_1^2=A_2^2=\1$ the unitary operation $\e^{i \theta A_1\otimes A_2}$
	can be expressed via the quasi-probability decomposition
	\begin{equation}
		\begin{tikzcd}[column sep=2mm]
			  & \gate[2]{\e^{i \theta A_1\otimes A_2}} &\qw \\
			  & \qw & \qw 
		\end{tikzcd}
	= \sum_{k=1}^{10} c_i \left(
	\begin{tikzcd}[column sep=2mm]
		& \gate{P_k}& \qw\\
		& \gate{Q_k} & \qw
	\end{tikzcd}\right)
	\end{equation}
with
\begin{equation}
\begin{aligned}
	c_1=&\cos^2 \theta\\
	c_2=&\sin^2 \theta\\
	c_k=& \cos\theta \sin\theta\quad \forall k\in \{3,6,7,10\}\\
	c_k=& -\cos\theta \sin\theta\quad \forall k\in \{4,5,8,9\}
\end{aligned}
\end{equation}
and
\begin{equation}
\begin{aligned}
	P_1 =& \1, & Q_1 =& \1\\
	P_2 =& A_1,& Q_2 =& A_2\\
	P_3 =& \frac{1}{2}\left(\1+A_1\right), & Q_3 =& \frac{1}{\sqrt{2}}\left(\1+i A_2\right)\\ 
	P_4 =& \frac{1}{2}\left(\1-A_1\right), & Q_4 =& \frac{1}{\sqrt{2}}\left(\1+i A_2\right)\\
	P_5 =& \frac{1}{2}\left(\1+A_1\right), & Q_5 =& \frac{1}{\sqrt{2}}\left(\1-i A_2\right)\\
	P_6 =& \frac{1}{2}\left(\1-A_1\right), & Q_6 =& \frac{1}{\sqrt{2}}\left(\1-i A_2\right)\\
	P_7 =& \frac{1}{\sqrt{2}}\left(\1+i A_1\right), & Q_7 =& \frac{1}{2}\left(\1+ A_2\right)\\
	P_8 =& \frac{1}{\sqrt{2}}\left(\1-i A_1\right), & Q_8 =& \frac{1}{2}\left(\1+ A_2\right)\\
	P_9 =& \frac{1}{\sqrt{2}}\left(\1+i A_1\right), & Q_9 =& \frac{1}{2}\left(\1- A_2\right)\\
	P_{10} =& \frac{1}{\sqrt{2}}\left(\1-i A_1\right), &\quad Q_{10} =& \frac{1}{2}\left(\1- A_2\right).
\end{aligned}
\end{equation}
The factors $\frac{1}{2}(\1\pm A_{1/2})$ and $\frac{1}{\sqrt{2}}(\1\pm i A_{1/2})$ correspond to projective measurements and rotations, respectively. 
\end{lemma}
Lemma~\ref{lemma:horizontalcut} implies that a two-qubit gate can be simulated by six circuits where the gate is replaced by local operations, i.e. we can ``cut'' the gate horizontally.
The expectation value of any observable $\mathcal{A}$ on the state produced by the original circuit on the initial state $\rho_0$, can be obtained from the quasi-probability decomposition as
\begin{equation}
	\begin{aligned}
	\langle\mathcal{A}\rangle =& \mathrm{tr}\left(\mathcal{A} \e^{i \theta A_1\otimes A_2}\rho_0\e^{-i \theta A_1\otimes A_2}\right)\\
	=& \sum_{k=1}^{10} c_k \mathrm{tr}\left(\mathcal{A} (P_k\otimes Q_k)\rho_0(P_k^\dagger\otimes Q_k^\dagger)\right)
	\end{aligned}
\end{equation}
due to the linearity of the trace.

We now consider timewise cuts. 
They introduce two new boundaries on which we might apply the simplification rules discussed above. 
\begin{lemma}[Vertical cut~\cite{Peng2020}]\label{lemma:verticalcut}
The identity operation on a single qubit can be decomposed as
\begin{equation}
	\raisebox{4pt}{\begin{tikzcd}[column sep=2mm]
			&[1cm] \qw & 
	\end{tikzcd}}
	=
	\sum_{k=1}^8 d_k \left(
	\raisebox{6pt}{\begin{tikzcd}[column sep=2mm]
		&[0.25cm] \meter{$O_k$} &  \rho_k &[0.25cm] \qw 
	\end{tikzcd}}\right),
\end{equation}
where
\begin{alignat}{2}
O_k &= \sigma_{\lfloor (k-1)/2\rfloor},\\
\rho_1 &= \proj{0}\nonumber\\
\rho_2 &= \proj{1}\\
\rho_k &= \frac{1}{2}\left(\1 - (-1)^k O_k\right) \quad \forall k>2,\nonumber\\
\intertext{and} d_k &= \left\{\begin{array}{cl}
	-\frac{1}{2} & \text{if } k\in \{4,6,8\}\\
	\frac{1}{2} & \text{else.}
\end{array}\right.
\end{alignat}
The symbol $\sigma_k\in \{\1,X,Y,Z\}$ denotes the identity and the Pauli matrices.
\end{lemma}
By cutting (few) wires of the quantum circuit we can simulate a larger quantum circuit with smaller ones. 
The sum in Lemma~\ref{lemma:verticalcut} dictates how the measurement statistics of the original circuit is obtained from the measurements on the smaller circuits.
Note that the states which need to be prepared, given by the density matrices $\rho_i$, are pure states. 
In practice, Lemma~\ref{lemma:verticalcut} suffers from the caveats of noisy state tomography with finite sample size, e.g. negative probabilities in the state. 
We neglect those issues here, but refer the reader to the extensive literature, e.g.~\cite{Christandl2012, Schwemmer2015, Perlin2021}.

The following example illustrates how circuit cutting can allow to apply the described hybrid simplification rules to more locations in a quantum circuit.
\begin{example}[Rule~\ref{rule:incoherent} after a cut]
	We consider circuits which are almost in the form of Rule~\ref{rule:incoherent}, where circuit cutting allows us to apply the rule.
	\begin{enumerate}
		\item[(a)] In the left hand side of
\begin{equation}
	\begin{aligned}
	&\begin{tikzcd}[column sep=2mm]
		\lstick{...}  & \phase[label position=below]{\text{\raisebox{-5pt}{\ScissorRight}}\hspace{0.25cm}}& \qw & \qw \rstick{...}\\[-0.5cm]
		\lstick{$\ket{0}$}& \targ{}\vqw{-1} & \gate{U_i} &\meter{}\vcw{1}&\\
		\lstick{...}&\qwbundle[alternate]{}&\ctrlbundle{-1} & \gate{F_\mu}\qwbundle[alternate]{} & \meter{}\qwbundle[alternate]{}
	\end{tikzcd}\\
	=&\sum_k c_k \left\{
	\begin{tikzcd}[column sep=2mm]
		\lstick{...}  & \gate{P_k}& \qw & \qw \rstick{...}\\[-0.5cm]
		\lstick{$\ket{0}$}& \gate{Q_k} & \gate{U_i} &\meter{}\vcw{1}&\\
		\lstick{...}&\qwbundle[alternate]{}&\ctrlbundle{-1} & \gate{F_\mu}\qwbundle[alternate]{} & \meter{}\qwbundle[alternate]{}
	\end{tikzcd}\right.
\end{aligned}
\end{equation}
		the ancillary qubit interacts with another qubit via a controlled-not gate. 
		This prevents the direct application of Rule~\ref{rule:incoherent} to remove the ancillary qubit.
		A horizontal cut on the controlled-not gate replaces it with local operations.
		Each term of the quasi-probability decomposition now allows to apply Rule~\ref{rule:incoherent}. 
		The unitary operations on the ancillary qubit can be absorbed into the definition of $U_i$. 
		The projective measurement has no effect, because the initial state $\ket{0}$ of the ancillary qubit is already a computational basis state.
		\item[(b)] 
		In the left hand side of
		\begin{equation}
			\begin{aligned}
			&\begin{tikzcd}[column sep=2mm]
			&\gate[2]{V}&\push{\text{\ScissorRight}}&\gate{U_i} & \meter{}\vcw{1}&\\
			&\qwbundle[alternate]{}&\qwbundle[alternate]{}&\ctrlbundle{-1} & \gate{F_\mu}\qwbundle[alternate]{} & \meter{}\qwbundle[alternate]{}
			\end{tikzcd}\\
		=&\sum_k d_k\left\{ \begin{tikzcd}[column sep=2mm]
			&\gate[2]{V}&\meter{$O_k$} &  \rho_k &\qw&\gate{U_i} & \meter{}\vcw{1}&\\
			&\qwbundle[alternate]{}&\qwbundle[alternate]{}&\qwbundle[alternate]{}&\qwbundle[alternate]{}&\ctrlbundle{-1} & \gate{F_\mu}\qwbundle[alternate]{} & \meter{}\qwbundle[alternate]{}
		\end{tikzcd}\right.
			\end{aligned}
		\end{equation}
		the ancillary qubit interacts coherently with the system first. 
		This prevents us from computing the measurement statistics on the ancillary system via Eq.~(\ref{eq:Pmux}).
		A vertical cut at the indicated location separates the coherent from the incoherent interaction and allows us to apply Rule~\ref{rule:incoherent}. 
	\end{enumerate}
\end{example}
\section{Conclusion}\label{sec:conclusion}
In the present paper we formalized the intuition that incoherent operations can be pushed across the boundary between a quantum circuit and classical pre- or post-processing.
As a result we collected four simple rules which can be used to simplify quantum circuits.
Those rules are particularly useful when the quantum circuit runs on NISQ hardware.
When implemented in a compiler they can automatically simplify circuits and contemplate other simplification rules which work only on the quantum part of the calculation.
We illustrated how circuit cutting can extend the usefulness of the rules.
Automatic identification of beneficial circuit cuts should improve the results of compilers for quantum circuits. 
Hybrid simplification rules should be particularly useful in the context of hybrid algorithms, like variational quantum algorithms, which already contain classical blocks of computation.
Further studies are required to investigate how much the discussed rules can affect the performance of NISQ devices for such applications.
\begin{acknowledgements}
	This work was funded by the German Federal Ministry for Economic Affairs and Climate Action (BMWK) via the project AQUAS. 
\end{acknowledgements}
\bibliographystyle{unsrtnat}

\onecolumngrid
\clearpage
\appendix
\end{document}